\documentclass{article}[12pt]

\usepackage{amssymb, amsthm, amsmath,graphicx}
\usepackage{a4wide}

\newcommand{\fra}[2]{\textstyle{\frac{#1}{#2}}}
\newcommand{\beqn}{\begin{eqnarray}\begin{aligned}}
\newcommand{\eqn}{\end{aligned}\end{eqnarray}}

\usepackage{natbib}

\newtheorem{thm}{Theorem}[section]

\newtheorem{lem}[thm]{Lemma}

\usepackage{graphics,graphicx}
\usepackage{pstricks,pst-node,pst-tree}
\usepackage{braket}

\begin{document}

\begin{titlepage}

\begin{center}

{\Large {\bf The algebra of the general Markov model on phylogenetic trees and networks}}

\vspace{2em}

J. G. Sumner, B. H. Holland$^{\ast}$, and P. D. Jarvis$^{\dagger}$  
\par \vskip 1em \noindent
%{\it $^1$School of Information Technologies, $^3$School of Biological Sciences, $^4$Centre for Mathematical Biology, $^5$Sydney Bioinformatics, University of Sydney, NSW 2006, Australia}\\
{\it School of Mathematics and Physics, University of Tasmania, Australia}
%{\it $^3$School of Biological Sciences, University of Sydney}\\
%{\it $^4$Centre for Mathematical Biology}\\
%{\it $^5$Sydney Bioinformatics}\\
\end{center}
%\footnotetext{ }
\par \vskip .3in \noindent

\vspace{1cm} \noindent\textbf{Abstract} \normalfont 
\\\noindent
It is known that the Kimura 3ST model of  sequence evolution on phylogenetic trees can be
extended quite naturally to arbitrary split systems. 
However, this extension relies heavily on mathematical peculiarities of the K3ST model, and providing an analogous augmentation of the general Markov model has thus far been elusive.
In this paper we rectify this shortcoming by showing how to extend the general Markov model on trees to to include arbitrary
splits; and even further to more general network models.  
This is achieved by exploring the algebra of the generators of the continuous-time Markov chain together with the ``splitting'' operator that
generates the branching process on phylogenetic trees. 
For simplicity we proceed by discussing the two state case and note that our results are easily extended to more states with little complication.
Intriguingly, upon restriction of the two state general Markov model to the parameter space of the binary symmetric model, our extension is indistinguishable from the previous approach only on trees; as soon as any incompatible splits are introduced the two approaches give rise to differing probability distributions with disparate structure. 
Through exploration of a simple example, we give a tentative argument that our approach to extending to more general networks has desirable properties that the previous approaches do not share.
In particular, our construction allows for the possibility of convergent evolution of previously divergent lineages; a property that is of significant interest for biological applications.
\vfill
\hrule \mbox{} \\
{\footnotesize{
$^{\ast}$ ARC Future Fellow\\
$^{\dagger}$ Alexander von Humboldt Fellow\\
\textit{keywords:} split system, Markov process, maximum likelihood\\
\textit{email:} jsumner@utas.edu.au
}}
\end{titlepage}

%\tableofcontents

%
\section{Introduction}\label{sec:intro}
Phylogenetic methods seek to infer the prior evolutionary relationships of extant taxa.
Classically, this was achieved by comparing morphological features, but modern methods focus on molecular data such as DNA.
Harking back to sketches in Darwin's early notebooks, it has also been assumed that evolutionary history resembles a tree structure.
However, it is now well known that evolutionary processes such as hybridisation, deep coalescence (incomplete lineage-sorting), horizontal gene transfer and recombination cannot be accurately modelled as a tree. 
Even when the underlying historical signal fits a tree, there may be conflicting non-historical signals caused by sampling error, long-branch attraction, nucleotide composition bias, or changes in the substitution rate at individual sites across the tree, as well as alignment or misreading errors.  
Incorrect or over-parameterized models of sequence mutations can lead to high statistical support for splits that are incompatible with a single tree.

It is clear that imposing a strictly treelike evolutionary history may be inappropriate in the situations described above, hence methods that can assist in identifying and understanding conflict in phylogenetic data are essential. 
One class of methods that have proved useful in this respect are weighted split-systems and their corresponding visualization as networks. 
Split networks initially arose as means of visualizing the split decomposition of a distance metric as defined by \citet{bandelt1992}. 
In that work, the authors gave a decomposition that provides a way of assessing whether the structure of a distance matrix is treelike or if it contains other conflicting signals. 
However, the idea of a weighted split-system is very general and has arisen in many phylogenetic contexts. 
These include (1) median networks \citep{bandelt1994}, where splits and their weights are derived from a binary coding of a sequence alignment; (2) Hadamard (or spectral) analysis \citep{hendy1989}, which defines an invertible relationship between site patterns and a split spectrum under certain simple models (K3ST and subclasses); (3) Neighbor-Net \citep{bryant2004}, a distance-based method which applies a greedy agglomerative algorithm to find a circular ordering of taxa and then a least-squares approach to find weights for the corresponding set of circular splits; (4) Consensus networks \citep{holland2004}, which take a set of trees and define a weighted split-system based on the number of trees that display a particular edge, possibly incorporating edge weight information \citep{holland2006}.

With the exception of spectral analysis \cite{hendy1989} these methods are all combinatorial and/or distance-based; there is currently no way to infer a weighted split-system in the likelihood setting using general Markov models of sequence evolution. 
That said, there has been some previous work on calculating likelihood scores for particular phylogenetic networks under special models. 
\citet{haeseler1993} developed a framework for computing the likelihood of a split-system for binary sequences under a Cavender-Farris model with invariable sites. 
\citet{strimmer2000} developed a Bayesian approach that calculated the likelihood of a given directed acyclic graph (DAG) for more complex models of sequence evolution. 
More recently \citet{jin2006} defined a likelihood score for phylogenetic networks as a weighted mixture of tree likelihoods. 
All of these methods begin with a given phylogenetic network and then attempt to calculate its likelihood under some model. 
This is very different from the Hadamard based approach -- and the new approach we explore here -- which begin with the data and a model and infer a weighted split-system.

Given their importance, there is a distinct lack of an extension of the standard Markov models on trees to arbitrary split systems.
In recent work, \citet{bryant2005c,bryant2009} has re-examined the nature of the Kimura 3ST and binary-symmetric model as, under a simple extension, these models permit the inclusion of arbitrary splits over and above those that come from a single tree. 
However, these so called ``group-based'' models of sequence evolution are not motivated by biological considerations and hence their validity in applied studies must be scrutinized carefully. 
The primary motivation for these models is mathematical elegance and simplicity, so that it is not necessarily the case that the underlying assumptions have biological relevance.
Specifically, all group-based models have doubly-stochastic rate matrices and thus uniform stationary distributions.
This is clearly inappropriate given that varying GC content is known to be of crucial importance in phylogenetics \citep{jermiin2004}.
It appears that to date it has not been possible to employ general model-based methods to infer split networks from phylogenetic data sets.

In this article we show how Markov models of phylogenetic evolution on trees (thought of as \emph{compatible} split systems) can be generalized to the case of \emph{arbitrary} split systems.
In the binary case of two character states, we achieve this by studying the algebra of the generators of the continuous time Markov chain together with the ``splitting'' operator that generates the branching process on phylogenetic trees.
The resulting presentation of the general Markov rate matrix model on a tree is such that it can be generalized in a natural way to include arbitrary splits; including those that are incompatible with any tree.
This results in a very general model that contains the standard tree model as a special case, but has the potential to associate an individual weight and rate matrix to \emph{any} additional splits that we wish to include.
Additionally, we show by example that our approach gives rise to the possibility of Markov models on much more general networks, with phylogenetic evolution proceeding in a series of ``epochs'' consisting of divergence or \emph{convergence} of arbitrary groups of taxa (ie. lineages).
As part of the discussion, we note that our results are fully generalizable to any number of character states with complication of detail only.
Intriguingly, we will also show that under a restriction of the parameter space to the binary-symmetric case that this model is \textit{not} consistent with the Hadamard based approach given by \citet{bryant2009}.
We close with a simple example that shows our construction has the ability to model convergent evolution of lineages; a property that is simply not available to the Hadamard based approach.

\section{Preliminaries}\label{sec:prelim}

In this article, a \emph{tree} with vertex set $X$ is an acyclic graph with vertices chosen from $X$ such that all vertices have valence of exactly 3 or 1. 
A \emph{rooted tree} is an acyclic graph as above but with a single vertex $\rho$ (the root) having valence 2.
Thus, a rooted tree is a collection of \emph{edges} $e\in \binom{X}{2}$, and can be made to be a \emph{directed} graph by considering $e=(u,v)$ as an ordered pair of adjacent vertices, where $u$ lies on the path from $v$ to $\rho$.
Vertices of valence 1 are referred to as \emph{leaves} and we label the leaves from elements of the set $\left[n\right]:=\{1,2,\ldots,n\}$.
For a rooted tree, we label each non-leaf vertex $v$ by the subset of leaves such that the path from each of these leaves to $\rho$ contains $v$. 
Additionally, we label each edge $e=(u,v)$ by the subset that labels the vertex $v$.
In this way, the edges are labelled by subsets of $\left[n\right]$, with pendant edges labelled by singletons.
See Figure~\ref{fig:tree} for a graphical representation of a rooted tree.

\begin{figure}[t]
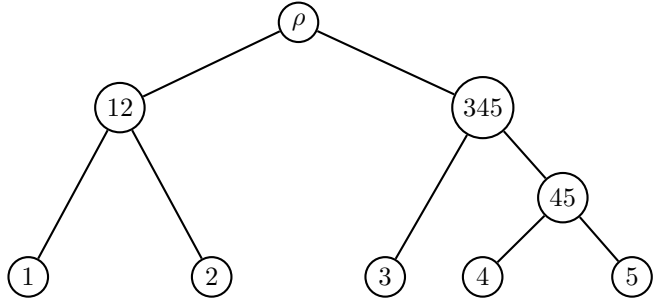

  \centering  
  $\psmatrix[colsep=.3cm,rowsep=.4cm,mnode=circle]
  	&&&&&& \rho \\
   && 12 &&   && &&  && 345 \\ 
   &&  &&   && &&  &&  & 45 \\
  1 && &&  2 && && 3 && 4 && 5 
  \ncline{1,7}{2,3}
  %\taput{12}
  \ncline{2,3}{4,1}
  %\taput{1}
  \ncline{2,3}{4,5}
  %\taput{2}
  \ncline{1,7}{2,11}
  %\taput{345}
  \ncline{2,11}{3,12}
  %\taput{45}
  \ncline{2,11}{4,9}
  %\taput{3}
  \ncline{3,12}{4,13}
  %\taput{5}
  \ncline{3,12}{4,11}
  %\taput{4}
  \endpsmatrix
  $
  \caption{A rooted tree.}
  \label{fig:tree}
\end{figure}

Consider the vector space $V\cong \mathbb{C}^2$ with the ordered basis\footnote{We use ``Dirac notation'', where a vector is represented by a ``ket'' $\Ket{\mbox{}}$, as this notation is particularly elegant when it comes to more general phylogenetic character patterns.}
\beqn
\left\{\Ket{0}\equiv e_0:=\left(
\begin{array}{r}
	1\\
	0\\
\end{array}\right),
\hspace{1em}
\Ket{1}\equiv e_1:=\left(
\begin{array}{r}
	0\\
	1\\
\end{array}\right)\right\}.\nonumber
\eqn
With respect to this basis, we can define ``Markov generators'' as the zero column-sum matrices
\beqn
L_{\alpha}=\left(
\begin{array}{rr}
	-1 & 0 \\
	1  & 0 \\
\end{array}\right),
\qquad
L_{\beta}=\left(
\begin{array}{rr}
	0 & 1 \\
	0  & -1 \\
\end{array}\right).\nonumber
\eqn
In this way, the general rate matrix for a continuous-time Markov chain on two states can be expressed as the linear combination\footnote{An amusing aside: Our rate matrices have zero \emph{column-} rather than \emph{row}-sum, as we like to conform with physicists' notation of right matrix multiplication. This sits well psychologically with the physical picture that a linear operator ``hits'' a vector and it ``moves'', which, in turn, is in tune with the left-to-right direction that one reads printed English.}
\beqn\label{eq:ratematrix}
Q=\alpha L_{\alpha}+\beta L_{\beta}=\left(
\begin{array}{rr}
	-\alpha & \beta \\
	\alpha & -\beta  \\
\end{array}\right).
\eqn
The associated transition matrix $\left[M(t)\right]_{ij}$, representing the probability of a transistion $j\rightarrow i$ at time $t$, is then given by the exponential map
\beqn
M(t)=\exp\left[{Qt}\right]:=\sum_{n=0}^\infty \frac{Q^nt^n}{n!}.\nonumber
\eqn

By noting that the generators satisfy the relations 
\beqn\label{eq:relations}
L_{\alpha}^2=L_{\beta}L_{\alpha}=-L_{\alpha},\qquad L_{\beta}^2=L_{\alpha}L_{\beta}=-L_{\beta},
\eqn
it is easy to show that 
\beqn
\left(\alpha L_{\alpha}+\beta L_{\beta}\right)^n=(-1)^{n-1}(\alpha+\beta)^{n-1}\left(\alpha L_{\alpha}+\beta L_{\beta}\right).\nonumber
\eqn
Thus
\beqn\label{eq:transmatrix}
M(t)=\exp\left[{Qt}\right]=\sum_{n=0}^\infty \frac{Q^nt^n}{n!}&=1-\frac{1}{\alpha+\beta}\left(e^{-(\alpha+\beta)t}-1\right)\left(\alpha L_{\alpha}+\beta L_{\beta}\right)\\
&=1-\frac{1}{\alpha+\beta}\left(e^{-(\alpha+\beta)t}-1\right)Q.
\eqn

As $M(t)$ is invariant under the reparameterization\footnote{For further discussion of \emph{local} time-reparameterization in phylogenetics see \citet{jarvis2005} or, in the context of a changing rate of mutation, see \citet{penny2005}}
\beqn
t &\rightarrow t'=\lambda t,\quad
\alpha &\rightarrow \alpha'=\lambda^{-1}\alpha,\quad
\beta &\rightarrow \beta'=\lambda^{-1}\beta,\nonumber 
\eqn
we see that we can ``scale out'' $t$ by choosing $\lambda\!=\!t^{-1}$.
As, in a practical context, $\alpha$, $\beta$ and even $t$ are unknown parameters that must be inferred from observed data using some statistical estimation procedure, we see that we can take
\beqn
M(\alpha,\beta)=e^{Q}=e^{\left(\alpha L_\alpha+\beta L_\beta\right)},\nonumber
\eqn
as completely equivalent to (\ref{eq:transmatrix}).  
If we think of $M(\alpha,\beta)$ as a two-dimensional manifold (in the sense of a Lie group \citep{procesi2007}), then we see that the Markov generators are none other than the basis vectors of the tangent space at the identity:
\beqn
L_{\alpha}\equiv \left.\frac{\partial}{\partial \alpha}M(\alpha,\beta)\right|_{\alpha=\beta=0},\qquad L_{\beta}\equiv \left.\frac{\partial}{\partial \beta}M(\alpha,\beta)\right|_{\alpha=\beta=0},\nonumber
\eqn
with algebraic closure across the ``Lie bracket'' $\left[L_\alpha,L_\beta\right]:=L_\alpha L_\beta-L_\beta L_\alpha=L_\alpha-L_\beta$ ensuring ``closure'' of the corresponding Markov model (as is discussed in \cite{jarvis2010}).
This connection between continuous time Markov chains and Lie groups is an important one and seems to have been first noted by \citet{johnson1985}.
This point of view is needed in order to extend the results of the present article to the case of character state spaces of arbitrary size.
Having given this perspective into the meaning of the Markov generators, we from will nevertheless take the more usual representation (\ref{eq:transmatrix}) of transition matrices in all that follows below.

In \citet{sumner2005} it was shown that the Markov models of phylogenetics in standard use can be represented in an abstract setting using the tensor product space $V\otimes V\otimes \ldots \otimes V$, where $\dim(V)\!=\!k$ is the number of character states and the number of copies of $V$ is equal to the number of taxa under consideration.
In these models, it is usual to impose conditional independence across the branches of the tree and this can be formalized using a linear operator $\delta: V\rightarrow V\otimes V$ to generate speciation events. 
This is referred to as the ``splitting operator'' and is defined, using our chosen basis, as
\beqn
\delta\cdot \Ket{i} = \Ket{i}\otimes \Ket{i},\nonumber
\eqn
which, expressed at the level of the probability distributions, corresponds exactly to the duplication of a sequence of molecular units.
That is, if we select state $i$ from the initial sequence with probability $p_i$, then immediately after duplication -- and assuming the sequences remain aligned -- the probability of observing the pattern $ij$ at a given site is $p_{i}\delta_{ij}$.
By defining the vector $p:=\sum_{i}p_i\Ket{i}\in V$ and noting that the tensors $\left\{\Ket{ij}:=\Ket{i}\otimes \Ket{j}\right\}_{0\leq i,j \leq 1}$ form a basis for $V\otimes V$, the splitting operator achieves this notion of speciation in the abstract setting:
\beqn
\delta\cdot p = \delta \cdot \left(\sum_{i}p_i\Ket{i}\right) = \sum_{i} p_i \delta \cdot \Ket{i} = \sum_{i} p_i \Ket{ii} = \sum_{i,j} p_{i}\delta_{ij}\Ket{ij},\nonumber
\eqn
where $P:=\delta\cdot p\in V\otimes V$ has components $p_{i}\delta_{ij}$ and is referred to as a ``phylogenetic tensor''.
Subsequently, given two rate matrices $Q_1$ and $Q_2$, and two edge weights $\tau_1$ and $\tau_2$, the phylogenetic tensor evolves to
\beqn
P'=e^{Q_1\tau_1}\otimes e^{Q_2\tau_2}\cdot P,\nonumber
\eqn
which, up to first order terms in the edge weights, is
\beqn
P'&=\left[\left(1+\alpha_1\tau_1L_{\alpha}+\beta_1\tau_1L_{\beta}+\ldots\right)\otimes \left(1+\alpha_2\tau_2L_{\alpha}+\beta_2\tau_2L_{\beta}+\ldots\right)\right]\cdot P\nonumber\\
&=\left[1+\alpha_1\tau_1L_{\alpha}\otimes 1+\alpha_2\tau_2 1\otimes L_{\alpha} + \beta_1\tau_1L_{\beta}\otimes 1+\beta_2\tau_2 1\otimes L_{\beta}+\ldots  \right]\cdot P.
\eqn
In this way, the splitting operator can be thought of as the generator of the branching pattern of the phylogenetic tree, while $L_{\alpha}$ and $L_\beta$ are the generators of the Markov process.
(For more details of this formalism see \citet{bashford2004,sumner2005} and \citet{sumner2008}, and for an even more general setting see \citet{jarvis2005}).
Presently, we are concerned with the algebra resulting from application of these two types of generators.

\section{Some helpful lemmas}

The action of the Markov generators on the basis vectors is
\beqn\label{eq:basisaction}
\begin{array}{ll}
L_\alpha \Ket{0}=\Ket{1}-\Ket{0}, \hspace{1em}&\hspace{1em} L_\beta \Ket{0}=0, \\
L_\alpha \Ket{1}=0; \hspace{1em}&\hspace{1em} L_\beta \Ket{1}=\Ket{0}-\Ket{1}.
\end{array}
\eqn
By comparing the relations 
\beqn
\delta\cdot L_{\alpha}\Ket{0}&=\Ket{11}-\Ket{00},\qquad
\delta\cdot L_{\alpha}\Ket{1}&=0;\nonumber
\eqn
and
\beqn
L_{\alpha}\otimes L_{\alpha} \Ket{00}&=\Ket{11}-\Ket{01}-\Ket{10}+\Ket{00},\\
L_{\alpha}\otimes L_{\alpha} \Ket{11}&=0,\\
L_{\alpha}\otimes 1 \Ket{00}&=\Ket{10}-\Ket{00},\\
L_{\alpha}\otimes 1 \Ket{11}&=0,\\
1\otimes L_{\alpha} \Ket{00}&=\Ket{01}-\Ket{00},\\
1\otimes L_{\alpha} \Ket{11}&=0;\nonumber
\eqn
with similar for $L_{\beta}$, we find that the Markov generators can be ``pushed through'' past the splitting operator:
\begin{lem}\label{lem:intertwining}
As operators from $V$ to $V\otimes V$, we have
\beqn
\delta\cdot L_{\alpha}&=\left(L_{\alpha}\otimes L_{\alpha}+L_{\alpha}\otimes 1+1\otimes L_{\alpha}\right)\cdot\delta,\\
\delta\cdot L_{\beta}&=\left(L_{\beta}\otimes L_{\beta}+L_{\beta}\otimes 1+1\otimes L_{\beta}\right)\cdot\delta.\nonumber
\eqn
\end{lem}
In the terminology of group actions \citep{procesi2007}, this lemma tells us the rule for how the two operators ``intertwine''.
It is exactly this relation that we will exploit to show how to generalize from a Markov model on a tree to a model on a general split system.

Given a linear operator $X$ on $V$, we define a linear operator $X^{(i)}$ on $V\otimes V\otimes \ldots \otimes V$ as the tensor product
\beqn
X^{(i)}:=1\otimes 1\otimes \ldots \otimes X \otimes 1 \otimes \ldots \otimes 1,\nonumber
\eqn
where $X$ appears in the $i^{th}$ slot of the tensor product.
Further, for a subset $A\subseteq \left[n\right]:=\{1,2,\ldots,n\}$, we define
\beqn
X^{(A)}:=\prod_{i\in A}X^{(i)}.\nonumber
\eqn
For example, if we take $n=5$, we have
\beqn
X^{\left(25\right)}=\left(1\otimes X\otimes 1\otimes 1\otimes 1\right)\cdot \left(1\otimes 1\otimes 1\otimes 1 \otimes X\right)=1\otimes X\otimes 1\otimes 1 \otimes X.\nonumber
\eqn
Presently we will show how the interaction of $\delta$ with $L_{\alpha}$ naturally produces terms such as $L_{\alpha}^{(A)}$ (and similar for $L_{\beta}$).

\begin{lem}
As linear operators from $V$ to $V\otimes V\otimes V$,
\beqn
1\otimes \delta\cdot \delta = \delta \otimes 1\cdot \delta.\nonumber
\eqn
\end{lem}
\begin{proof}
We have
\beqn
1\otimes \delta\cdot \delta \cdot \Ket{i}=1\otimes \delta \cdot \left(\Ket{i}\otimes \Ket{i}\right) = \Ket{i} \otimes \left(\Ket{i}\otimes \Ket{i}\right) =\Ket{i}\otimes \Ket{i}\otimes \Ket{i}:=\Ket{iii}.\nonumber
\eqn
Similarly,
\beqn
\delta\otimes 1\cdot \delta \cdot \Ket{i} = \delta\otimes 1\cdot \left(\Ket{i}\otimes \Ket{i}\right) = \left(\Ket{i}\otimes \Ket{i}\right)\otimes \Ket{i} =\Ket{i}\otimes \Ket{i}\otimes \Ket{i}=\Ket{iii} .\nonumber
\eqn
\end{proof}
Using this lemma, we can recursively define 
\beqn
\delta^{i+1}:&= \delta \otimes 1 \otimes 1 \otimes \ldots 1\cdot \delta^{i} 
\equiv 1\otimes \delta \otimes 1\otimes \ldots \otimes 1\cdot \delta^{i}
 \equiv \ldots 
 \equiv 1\otimes 1\otimes \ldots \otimes 1\otimes \delta \cdot \delta^i,\nonumber
\eqn
with $\delta^{1}:=\delta$.
The action of the operator $\delta^{n-1}$ taking $V$ to $V\otimes V\otimes \ldots \otimes V$ generates exactly the ``$n$-taxon process'' as defined in \citet{bryant2009}, which, in turn, is completely equivalent to the formalism given in \citet{bashford2004}. 

If we note that $\delta \cdot 1=1\otimes 1\cdot \delta$ and consider Lemma~\ref{lem:intertwining}, we see that, for $x=\alpha$ or $\beta$, we have
\beqn
\delta^2\cdot L_{x}:&=\delta\otimes 1\cdot \delta \cdot L_{x}\\
&=\delta\otimes 1\cdot \left(L_{x}\otimes L_{x}+1\otimes L_{x}+L_{x}\otimes 1\right)\cdot \delta\\
&=\left[\left(L_{x}\otimes L_{x}\otimes L_{x}+L_{x}\otimes 1\otimes L_{x}+1\otimes L_{x}\otimes L_{x}\right)+\right.\\
&\hspace{5em}\left.1\otimes 1\otimes L_{x}+\left(L_{x}\otimes L_{x}\otimes 1+L_{x}\otimes 1\otimes 1+1\otimes L_{x}\otimes 1\right)\right]\cdot \delta^2\\
&=\left(\sum_{A\subseteq \{1,2,3\},A\neq \emptyset}L_x^{(A)}\right)\cdot \delta^2.\nonumber
\eqn
Generalizing this result we have:
\begin{lem}\label{lem:higherintertwine}
\beqn
\delta^{n-1}\cdot L_{x}=\left(\sum_{A\subseteq \left[n\right],A\neq \emptyset}L_{x}^{(A)}\right)\cdot \delta^{n-1}.
\eqn
\end{lem}
\begin{proof}
The proof is by induction. We have shown that the result is true for $n=3$. Assuming that it is true for some $n>3$, we have
\beqn
\delta^{n}\cdot L_{x} &=\left(\delta\otimes 1\otimes\ldots \otimes 1 \right)\cdot \left(\sum_{A\subseteq\left[n\right],A\neq \emptyset}L_{x}^{(A)}\right)\cdot\delta^{n-1}\nonumber\\
&=\left(\delta\otimes 1\otimes\ldots \otimes 1 \right)\cdot\left(\sum_{A\subseteq\left[n-1\right],A\neq \emptyset}1\otimes L_{x}^{(A)}+\sum_{A\subseteq\left[n-1\right]}L_{x}\otimes L_{x}^{(A)}\right)\cdot\delta^{n-1} \\
&= \left(\sum_{A\subseteq\left[n-1\right],A\neq \emptyset} 1\otimes 1\otimes L_{x}^{(A)} +\right.\\ &\hspace{5em}\left.\sum_{A\subseteq\left[n-1\right]}\left(L_{x}\otimes L_{x}\otimes L_{x}^{(A)}+L_{x}\otimes 1\otimes L_{x}^{(A)}+1\otimes  L_{x}\otimes L_{x}^{(A)} \right) \right)\cdot\delta^{n} \\
&=\left(\sum_{A\subseteq \left[n+1\right],A\neq \emptyset}L_{x}^{(A)}\right)\cdot\delta^{n}.\nonumber
\eqn
\end{proof}

For the two state model, recall that for $n\geq 1$ we have
\beqn
Q^n=\left(\alpha L_{\alpha}+\beta L_{\beta}\right)^n=(-1)^{n-1}(\alpha+\beta)^{n-1}\left(\alpha L_{\alpha}+\beta L_{\beta}\right).\nonumber
\eqn
If we define
\beqn\label{def:higherorder}
\mathfrak{L}_{x}^{\left[n\right]}:=\sum_{A\subseteq \left[n\right],A\neq \emptyset}L_{x}^{(A)},
\eqn
then Lemma~\ref{lem:higherintertwine} implies that we have the intertwining
\beqn
\delta^{n-1}\cdot L_{x}&=\mathfrak{L}_{x}^{\left[n\right]}\cdot \delta^{n-1}.\nonumber
\eqn
We also note that we have the recursion
\beqn
\mathfrak{L}_{x}^{\left[n\right]}=L_{x}\otimes \mathfrak{L}_{x}^{\left[n-1\right]}+1\otimes \mathfrak{L}_{x}^{\left[n-1\right]}+L_{x}\otimes 1\otimes 1\otimes  \ldots \otimes 1,\nonumber  
\eqn
and (after a little effort) it follows by induction on $n$ that
\beqn
\left(\mathfrak{L}_{\alpha}^{\left[n\right]}\right)^2=\mathfrak{L}_{\beta}^{\left[n\right]}\mathfrak{L}_{\alpha}^{\left[n\right]}=-\mathfrak{L}_{\alpha}^{\left[n\right]},\qquad \left(\mathfrak{L}_{\beta}^{\left[n\right]}\right)^2=\mathfrak{L}_{\alpha}^{\left[n\right]}\mathfrak{L}_{\beta}^{\left[n\right]}=-\mathfrak{L}_{\beta}^{\left[n\right]}.\nonumber
\eqn
Inspection reveals that, for all $n$, $\mathfrak{L}_{\alpha}^{\left[n\right]}$ and $\mathfrak{L}_{\beta}^{\left[n\right]}$ satisfy exactly the same algebraic relations as $L_{\alpha}$ and $L_{\beta}$ that were given in (\ref{eq:relations}).
It follows immediately that, for $n\geq 1$ we have
\beqn
\left(\alpha \mathfrak{L}_{\alpha}^{\left[n\right]}+\beta \mathfrak{L}_{\beta}^{\left[n\right]}\right)^n=(-1)^{n-1}(\alpha+\beta)^{n-1}\left(\alpha \mathfrak{L}_{\alpha}^{\left[n\right]}+\beta \mathfrak{L}_{\beta}^{\left[n\right]}\right),\nonumber
\eqn
and 
\beqn
\delta^{n-1}\cdot \left(\alpha L_{\alpha}+\beta L_{\beta}\right)^n&=(-1)^{n-1}(\alpha+\beta)^{n-1}\delta^{n-1} \cdot\left(\alpha L_{\alpha}+\beta L_{\beta}\right)\nonumber\\
&=(-1)^{n-1}(\alpha+\beta)^{n-1}\left(\alpha \mathfrak{L}_{\alpha}^{\left[n\right]}+\beta \mathfrak{L}_{\beta}^{\left[n\right]}\right)\cdot \delta^{n-1}\\
&=\left(\alpha \mathfrak{L}_{\alpha}^{\left[n\right]}+\beta \mathfrak{L}_{\beta}^{\left[n\right]}\right)^n\cdot \delta^{n-1}.
\eqn
Putting this together we see that 
\begin{lem}\label{lem:intertwineExp}
\beqn
\delta^{n-1}\cdot \exp\left[\alpha L_{\alpha}+\beta L_{\beta}\right]=\exp\left[\alpha \mathfrak{L}_{\alpha}^{\left[n\right]}+\beta \mathfrak{L}_{\beta}^{\left[n\right]}\right]\cdot \delta^{n-1}.\nonumber
\eqn
\end{lem}

In order to put all of the above to work, we require one final lemma regarding tensor products and the exponential map:
\begin{lem}\label{lem:ExpProduct}
Given any two linear operators $X$ and $Y$, we have
\beqn
e^{X}\otimes e^{Y}=e^{X\otimes 1+1\otimes Y}.\nonumber
\eqn
\end{lem}
\begin{proof}
Consider
\beqn
e^X\otimes 1= \left(1+X+\fra{1}{2}X^{2}+\ldots\right)\otimes 1=1\otimes 1+X\otimes 1+\fra{1}{2}\left(X\otimes 1\right)^2+\ldots =e^{X\otimes 1},\nonumber
\eqn
which implies that
\beqn
e^{X}\otimes e^{Y}=e^{X}\otimes 1\cdot 1\otimes e^{Y}=e^{X\otimes 1}\cdot e^{1\otimes Y}=e^{X\otimes 1+1\otimes Y},\nonumber
\eqn
where the last identity follows because $X\otimes 1$ and $1\otimes Y$ commute.
\end{proof}

\section{Alternative presentation of Markov models on trees}

Consider the tree presented in Figure~\ref{fig:tree2}.
Suppose we are given a root distribution $\Ket{\pi}~:=\!~\!\sum_{i}\pi_i\Ket{i}$, a rate matrix $Q=\alpha L_{\alpha}+\beta L_{\beta}$, and edge weights $\tau_1,\tau_2,\tau_3,\tau_{34}$ and $\tau_{234}$.
The phylogenetic tensor corresponding to this tree can be generated as
\beqn
P&=e^{Q\tau_1}\otimes e^{Q\tau_2}\otimes e^{Q\tau_3}\otimes e^{Q\tau_4}\cdot 1\otimes 1\otimes  \delta \cdot 1\otimes 1\otimes e^{Q\tau_{34}}\cdot 1\otimes \delta \cdot 1\otimes e^{Q\tau_{234}}\cdot \delta  \cdot \Ket{\pi}.\nonumber
\eqn
If we set $P=\sum_{i,j,k,l}p_{ijkl}\Ket{ijkl}$ and interpret $p_{ijkl}$ as the probability of observing the pattern $ijkl$ at the leaves of the tree, we see that this tensor is equivalent to specifying a joint distribution in the normal way (see \citet{semple2003} for example).

\begin{figure}[tbp]
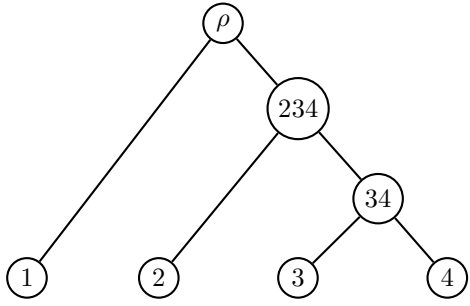

  \centering  
  $\psmatrix[colsep=.3cm,rowsep=.4cm,mnode=circle]
  && && &\rho\\
  && && && 234 \\
  && && && &34 \\
  1 && &&  2 && 3 && 4 
  \ncline{1,6}{4,1}
  \ncline{1,6}{2,7}
	\ncline{2,7}{4,5}
	\ncline{2,7}{3,8}
	\ncline{3,8}{4,9}
	\ncline{3,8}{4,7}
  \endpsmatrix
  $
  \caption{A rooted tree on four leaves.}
  \label{fig:tree2}
\end{figure}

By setting $Q=\alpha L_{\alpha}+\beta L_{\beta}$ and applying Lemma~\ref{lem:intertwineExp}, we find that 
\beqn
1\otimes 1\otimes \delta\cdot 1\otimes 1\otimes e^{Q\tau_{34}}=1\otimes 1\otimes \left(\delta\cdot e^{Q\tau_{34}}\right)=1\otimes 1\otimes e^{\left(\alpha \mathfrak{L}^{\left[2\right]}_{\alpha}+\beta \mathfrak{L}^{\left[2\right]}_{\beta}\right)\tau_{34}}\cdot 1\otimes 1\otimes \delta.\nonumber
\eqn
Now, we note that
\beqn
1\otimes 1\otimes \delta \cdot 1\otimes \delta=\left[1\otimes \left(1\otimes \delta\right)\right] \cdot 1\otimes \delta=1\otimes \left(1\otimes \delta \cdot \delta\right)=1\otimes \delta^2,\nonumber
\eqn
so again applying Lemma~\ref{lem:intertwineExp} we find that
\beqn
1\otimes \delta^2\cdot 1\otimes e^{Q\tau_{234}}=1\otimes e^{\left(\alpha \mathfrak{L}^{\left[3\right]}_{\alpha}+\beta \mathfrak{L}^{\left[3\right]}_{\beta}\right)\tau_{234}}\cdot 1\otimes \delta^2.\nonumber
\eqn
We also set
\beqn
1\otimes \delta^2\cdot \delta =\delta^{3}.\nonumber
\eqn
Thus 
\beqn
P=e^{Q\tau_1}\otimes e^{Q\tau_2}\otimes e^{Q\tau_3}\otimes e^{Q\tau_4}\cdot 1\otimes 1\otimes e^{\left(\alpha \mathfrak{L}^{\left[2\right]}_{\alpha}+\beta \mathfrak{L}^{\left[2\right]}_{\beta}\right)\tau_{34}}\cdot 1\otimes e^{\left(\alpha \mathfrak{L}^{\left[3\right]}_{\alpha}+\beta \mathfrak{L}^{\left[3\right]}_{\beta}\right)\tau_{234}} \cdot \Ket{\delta^3\pi},\nonumber
\eqn 
with $\Ket{\delta^3\pi}:=\delta^3\cdot \Ket{\pi}=p_0\Ket{0000}+p_1\Ket{1111}$.
Finally, by applying Lemma~\ref{lem:ExpProduct} multiple times, we find that we can write
\beqn\label{eq:expPhyloTensor}
P=\exp\left[\tau_1\mathcal{R}_1+\tau_2\mathcal{R}_2+\tau_3\mathcal{R}_3+\tau_4\mathcal{R}_4\right]\cdot\exp\left[\tau_{34}\mathcal{R}_{34}\right]\cdot\exp\left[\tau_{234}\mathcal{R}_{234}\right] \cdot \Ket{\delta^3\pi},
\eqn
where
\beqn
\mathcal{R}_1&=\alpha \left(L_{\alpha}\otimes 1\otimes 1\otimes 1\right)+\beta \left( L_{\beta}\otimes 1\otimes 1\otimes 1\right),\\
\mathcal{R}_2&=\alpha \left(1\otimes  L_{\alpha}\otimes 1\otimes 1\right)+ \beta \left(1\otimes  L_{\beta}\otimes 1\otimes 1\right),\\
\mathcal{R}_3&=\alpha \left(1\otimes 1\otimes L_{\alpha}\otimes 1\right)+\beta\left(1\otimes 1\otimes L_{\beta}\otimes 1\right),\\
\mathcal{R}_4&=\alpha \left(1\otimes 1\otimes 1\otimes L_{\alpha}\right)+\beta\left(1\otimes 1\otimes 1\otimes L_{\beta}\right),\\
\mathcal{R}_{34}&=\alpha \left(1\otimes 1\otimes \mathfrak{L}_{\alpha}^{\left[2\right]}\right)+\beta\left( 1\otimes 1\otimes \mathfrak{L}_{\beta}^{\left[2\right]}\right),\\
\mathcal{R}_{234}&=\alpha \left(1\otimes  \mathfrak{L}_{\alpha}^{\left[3\right]}\right)+\beta\left(1\otimes \mathfrak{L}_{\beta}^{\left[3\right]}\right).\nonumber
\eqn 

Suppose instead of the tree above we considered the quartet given in Figure~\ref{fig:tree3}.
The phylogenetic tensor corresponding to this tree is given by
\beqn
P'&=e^{Q\tau_1}\otimes e^{Q\tau_2}\otimes e^{Q\tau_3}\otimes e^{Q\tau_4}\cdot \delta \otimes \delta \cdot e^{Q\tau_{12}} \otimes e^{Q\tau_{34}} \cdot \delta  \cdot \Ket{\pi}.\nonumber
\eqn
By a similar argument to the one just given, it is possible to show that this tensor can be re-expressed as
\beqn%\label{eq:treeBalanced}
P'=\exp\left[\tau_1\mathcal{R}_1+\tau_2\mathcal{R}_2+\tau_3\mathcal{R}_3+\tau_4\mathcal{R}_4\right]\cdot\exp\left[\tau_{12}\mathcal{R}_{12}+\tau_{34}\mathcal{R}_{34}\right]\cdot \Ket{\delta^3\pi},\nonumber
\eqn
with 
\beqn
\mathcal{R}_{12}&=\alpha \left(\mathfrak{L}_{\alpha}^{\left[2\right]}\otimes 1\otimes 1\right)+\beta\left(\mathfrak{L}_{\beta}^{\left[2\right]}\otimes 1\otimes 1\right).\nonumber
\eqn

We can extend our definition (\ref{def:higherorder}) of $\mathfrak{L}^{\left[n\right]}_{x}$ to arbitrary subsets by taking
\beqn
\mathfrak{L}^{A}_x:=\sum_{A\subseteq B, A\neq \emptyset}L_{x}^{(A)},\nonumber
\eqn
for all $A\subseteq \left[n\right]$.
Now label the elements of $A$ as $A=\{a_1,a_2,\ldots ,a_{|A|}\}$ and consider a permutation $\sigma \in \mathfrak{S}_n$ such that $\sigma (a_i)=i$ (obviously such a permutation always exists). 
If we allow $\sigma$ to act on $V^{\otimes n}$ by permuting tensor factors, it is clear that \[\sigma \left(\mathfrak{L}^{A}_x\right)=\mathfrak{L}^{\left[|A|\right]}_x\otimes 1^{\left(\left[n-|A|\right]\right)}.\]
From this we can conclude that for fixed $A$ the operators $\left\{\mathfrak{L}_{\alpha}^{A},\mathfrak{L}_{\beta}^{A}\right\}$ also satisfy the exact same algebra as $\left\{L_\alpha,L_\beta\right\}$:
\beqn
\left(\mathfrak{L}_{\alpha}^{A}\right)^2=\mathfrak{L}_{\beta}^{A}\mathfrak{L}_{\alpha}^{A}=-\mathfrak{L}_{\alpha}^{A},\qquad \left(\mathfrak{L}_{\beta}^{A}\right)^2=\mathfrak{L}_{\alpha}^{A}\mathfrak{L}_{\beta}^{A}=-\mathfrak{L}_{\beta}^{A},\nonumber
\eqn 
for all $A\subseteq \left[n\right]$.
Using this result, we can unify the expressions for the rate matrices above by defining
\beqn
\mathcal{R}_{A}:=\alpha \mathfrak{L}^{A}_\alpha+\beta \mathfrak{L}^A_\beta.\nonumber
\eqn

%\begin{lem}
%$\left[\mathcal{R}_{A},\mathcal{R}_B\right]=0$ if and only if $A\cap B=\emptyset$.
%\end{lem}

Evidently $A\cap B$ implies that $\left[\mathcal{R}_{A},\mathcal{R}_B\right]=0$ and we see that there are several ways we can express our two phylogenetic tensors.
For instance the following presentations are all equally valid:
\beqn
P&=\exp\left[\tau_1\mathcal{R}_1+\tau_2\mathcal{R}_2+\tau_3\mathcal{R}_3+\tau_4\mathcal{R}_4\right]\cdot\exp\left[\tau_{34}\mathcal{R}_{34}\right]\cdot\exp\left[\tau_{234}\mathcal{R}_{234}\right] \cdot \Ket{\delta^3\pi}\nonumber\\
&=\exp\left[\tau_2\mathcal{R}_2+\tau_3\mathcal{R}_3+\tau_4\mathcal{R}_4\right]\cdot\exp\left[\tau_1\mathcal{R}_1+\tau_{34}\mathcal{R}_{34}\right]\cdot\exp\left[\tau_{234}\mathcal{R}_{234}\right] \cdot \Ket{\delta^3\pi}\\
&=\exp\left[\tau_2\mathcal{R}_2+\tau_3\mathcal{R}_3+\tau_4\mathcal{R}_4\right]\cdot\exp\left[\tau_{34}\mathcal{R}_{34}\right]\cdot\exp\left[\tau_1\mathcal{R}_1+\tau_{234}\mathcal{R}_{234}\right] \cdot \Ket{\delta^3\pi}\\
&=\exp\left[\tau_3\mathcal{R}_3+\tau_4\mathcal{R}_4\right]\cdot\exp\left[\tau_{34}\mathcal{R}_{34}\right]\cdot\exp\left[\tau_1\mathcal{R}_1+\tau_2\mathcal{R}_2+\tau_{234}\mathcal{R}_{234}\right] \cdot \Ket{\delta^3\pi}.
\eqn

%Given that we label edges of a rooted tree with $n$ leaves by subsets of $\left[n\right]$, it is clear that if $A,B\in T$ and $|A|=|B|$ then $A\cap B=\emptyset$.

\begin{figure}[tbp]
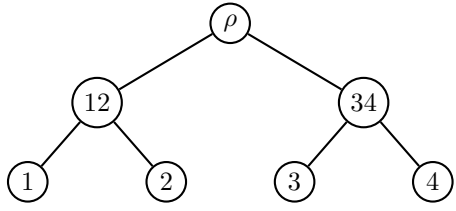

  \centering  
  $\psmatrix[colsep=.3cm,rowsep=.4cm,mnode=circle]
  && && \rho\\
  &&12 && && 34 \\
  &1& &2&  &3&  &4 
  \ncline{1,5}{2,3}
  \ncline{1,5}{2,7}
	\ncline{2,3}{3,2}
	\ncline{2,3}{3,4}
	\ncline{2,7}{3,6}
	\ncline{2,7}{3,8}
	 \endpsmatrix
  $
  \caption{Alternative tree on four leaves.}
  \label{fig:tree3}
\end{figure}
The content of our main theorem is that there is a canonical choice of presentation of a phylogenetic tensor for arbitrary trees.
Consider the presentations of the quartet tensors discussed above:
\beqn
P&=\exp\left[\tau_1\mathcal{R}_1+\tau_2\mathcal{R}_2+\tau_3\mathcal{R}_3+\tau_4\mathcal{R}_4\right]\cdot\exp\left[\tau_{34}\mathcal{R}_{34}\right]\cdot\exp\left[\tau_{234}\mathcal{R}_{234}\right] \cdot \Ket{\delta^3\pi},\\
P'&=\exp\left[\tau_1\mathcal{R}_1+\tau_2\mathcal{R}_2+\tau_3\mathcal{R}_3+\tau_4\mathcal{R}_4\right]\cdot\exp\left[\tau_{12}\mathcal{R}_{12}+\tau_{34}\mathcal{R}_{34}\right]\cdot \Ket{\delta^3\pi}.\nonumber
\eqn
\begin{thm}\label{thm:bigone}
Consider a rooted tree with $n$ leaves $\mathcal{T}=\{A_1,A_2,\ldots , A_{2n+2}\}$, where $A_i\subset \left[n\right]$.
Given a root distribution $\pi$, any rate parameters $\alpha$ and $\beta$ and weights $\{\tau_{A_1},\tau_{A_2},\ldots, \tau_{A_{2n+2}}\}$, a phylogenetic tensor (or joint distribution) $P$ at the leaves of this split system can be expressed as
\beqn
P=\exp\left[\mathcal{X}_1\right]\cdot\exp\left[\mathcal{X}_2\right]\cdot \ldots \cdot \exp\left[\mathcal{X}_{n-1}\right]\cdot  \Ket{\delta^{n-1} \pi},\nonumber
\eqn
with \[\Ket{\delta^{n-1} \pi}:=\delta^{n-1}\cdot \Ket{\pi}=p_0\Ket{00\ldots 0}+p_1\Ket{11\ldots 1},\] and \[\mathcal{X}_i=\sum_{A,|A|=i}\tau_{A}\mathcal{R}_{A}.\]
\end{thm}
\begin{proof}
The proof is by induction. 
Clearly a phylogenetic tensor on two leaves can be placed into the required form:
\[
P=\exp\left[\tau_1\mathcal{R}_1+\tau_{2}\mathcal{R}_2\right]\cdot \Ket{\delta \pi}.
\]
Assume that for $k>2$ any phylogenetic tensor on $k$ leaves $P^{(k)}$ can be expressed in the required form:
\[
P^{(k)}=\exp\left[\tau_{1}\mathcal{R}_{1}+\tau_{2}\mathcal{R}_{2}+\ldots +\tau_n\mathcal{R}_{n}\right]\cdot \exp\left[\sum_{A,|A|=2}\tau_{A}\mathcal{R}_{A}\right]\cdot \ldots \cdot \exp\left[\sum_{A,|A|=k-1}\tau_{A}\mathcal{R}_{A}\right]\cdot \Ket{\delta^{k-1}\pi}. 
\]
Without loss of generality we generate a phylogenetic tensor on $k+1$ leaves, $P^{(k+1)}$, by branching $P^{(k)}$ at the leaf $k$.
This is expressed algebraically by
\[
P^{(k)}\rightarrow P^{(k+1)}=\left(1\otimes 1\otimes \ldots \otimes 1\otimes \delta\right) \cdot P^{(k)}.
\]
For an arbitrary subset $A\subset \left[k\right]$, we have
\[
1\otimes 1\otimes \ldots \otimes 1\otimes \delta \cdot \mathcal{R}_{A}=\mathcal{R}_{A'}
\]
where
\[
A'=
\left\{
\begin{array}{l}
{A\cup \{k+1\}}\text{ , if }k\in A, \\
A\text{ , otherwise. }
\end{array}
\right.
\]
Pushing right through we have
\[
P^{(k+1)}=\exp\left[\sum_{A,|A|=1}\tau_{A}\mathcal{R}_{A'}\right]\cdot \exp\left[\sum_{A,|A|=2}\tau_{A}\mathcal{R}_{A'}\right]\cdot \ldots \cdot \exp\left[\sum_{A,|A|=n-1}\tau_{A}\mathcal{R}_{A'}\right]\cdot \Ket{\delta^{k}\pi}.
\]
Now consider, for a given $1\leq i < k-1$, the term
\[
\exp\left[\sum_{A,|A|=i}\tau_{A}\mathcal{R}_{A'}\right].
\]
As we are dealing with a tree $\mathcal{T}$, it is only possible  that one (and only one) of the subsets ($B_{i+1}$, say) in the summation has cardinality $i+1$.
Additionally, $B$ is disjoint from all of the other subsets, so we may write
\[
\exp\left[\sum_{A,|A|=i,A\neq \emptyset}\tau_{A}\mathcal{R}_{A'}\right]=\exp\left[\sum_{A,|A|=i,A\neq \emptyset,A\neq B}\tau_{A}\mathcal{R}_{A}\right]\exp\left[\tau_{B}\mathcal{R}_{B_{i+1}}\right].
\]    
Exactly the same argument is valid for the $i+1$ term: 
\[
\exp\left[\sum_{A,|A|=i+1,A\neq \emptyset}\tau_{A}\mathcal{R}_{A'}\right]=\exp\left[\sum_{A,|A|=i+1,A\neq \emptyset,A\neq B_{i+2}}\tau_{A}\mathcal{R}_{A}\right]\exp\left[\tau_{B}\mathcal{R}_{B+2}\right].
\]
Thus we can express the product of these two terms as
\beqn
\exp&\left[\sum_{A,|A|=i}\tau_{A}\mathcal{R}_{A'}\right]\cdot  \exp\left[\sum_{A,|A|=i+1}\tau_{A}\mathcal{R}_{A'}\right]\\
%&=\exp\left[\sum_{A,|A|=i,A\neq \emptyset,A\neq B}\tau_{A}\mathcal{R}_{A}\right]\exp\left[\tau_{B}\mathcal{R}_{B_{i+1}}\right]\exp\left[\sum_{A,|A|=i+1,A\neq \emptyset,A\neq B_{i+2}}\tau_{A}\mathcal{R}_{A}\right]\exp\left[\tau_{B}\mathcal{R}_{B+2}\right]\\
&=\exp\left[\sum_{A,|A|=i,A\neq \emptyset,A\neq B}\tau_{A}\mathcal{R}_{A}\right]\exp\left[\tau_B+\mathcal{R}_{B_{i+1}}+\sum_{A,|A|=i+1,A\neq \emptyset,A\neq B_{i+2}}\tau_{A}\mathcal{R}_{A}\right]\\
&\hspace{5em}\cdot\exp\left[\tau_{B}\mathcal{R}_{B+2}\right].\nonumber
\eqn
Continuing in this way we can place $P^{(k+1)}$ in the required form and the theorem follows by induction on $k$.
\end{proof}
Recall that, in the basis $\{\Ket{i_1i_2\ldots i_n}\}$, a phylogenetic tensor $P$ can be expressed as
\beqn
P=\sum_{i_1,i_2,\ldots i_n}p_{i_1i_2\ldots i_n}\Ket{i_1i_2\ldots i_n},\nonumber
\eqn
where $p_{i_1i_2\ldots i_n}$ is the probability of observing the pattern $i_1i_2\ldots i_n$ at the leaf vertices. 

From a biological perspective, it is apparent that the form given in Theorem~\ref{thm:bigone} that utilizes a cardinality ordering is somewhat mysterious.
However a little thought using commutivity (or otherwise) of the various operators shows that it is not so much the cardinality that matters, but it is that the operators that arise independently across branches of the tree are necessarily commutative and, conversely, those that do not commmute necessarily have non-zero intersection and hence are not independent.
Noting that there is some freedom in the final expression, we see that the cardinality ordering is simply a nice way of unifying the description for arbitrary trees. 

It is clear that if we take $P$ as in Theorem~\ref{thm:bigone} in the case that the split system $\mathcal{S}$ is compatible, we have a probability distribution on a tree identical to the standard presentation usually given in phylogenetics. 
However, the construction we have given naturally generalizes to a model on the most general split systems with trees occuring as sub-models (set the weights for incompatible splits equal to zero).
There is also an obvious generalization to the case where each split has a unique rate matrix -- simply give additional split labels to the rate parameters: $\alpha \rightarrow \alpha_{A}$ and $\beta \rightarrow \beta_{A}$.

In the next section we explore the $\alpha=\beta$ case in detail.

\section{(In)consistency with the Hadamard conjugation}\label{sec:incon}

In this section we consider the ``binary-symmetric'' case where $\alpha\!=\!\beta\!=\!\frac{1}{2}$, so that we can write
\beqn
Q=\frac{1}{2}\left(L_\alpha+L_\beta\right)=
\frac{1}{2}\left(
\begin{array}{rr}
	-1 & 1 \\
	1  & -1 \\
\end{array}\right)=\frac{1}{2}\left(-1+K\right),\nonumber
\eqn
where $K=\left(
\begin{array}{rr}
	0 & 1 \\
	1  & 0 \\
\end{array}\right)$ is the permutation matrix taking $\Ket{0}\rightleftharpoons \Ket{1}$.

For any tree $\mathcal{T}$ represented as a split system, it is shown in \citet{bashford2004} that we can write
\beqn
P=\exp\left[\sum_{x\in \mathcal{T}}w_x\frac{1}{2}\left(K^{(x)}-1\otimes 1\otimes \ldots \otimes 1\right)\right]\Ket{\delta^{n-1}\pi},\nonumber
\eqn
where $\{w_x\}_{x\in \mathcal{T}}$ is any set of edge weights on $\mathcal{T}$.
We will refer to this constructing of a phylogenetic tensor as the ``$K$-representation''.

On the other hand, we have shown in Theorem~\ref{thm:bigone} that for $\alpha\!=\!\beta\!=\!\frac{1}{2}$ we can write 
\beqn
P=\exp\left[\mathcal{X}_1\right]\cdot \exp\left[\mathcal{X}_2\right]\cdot \ldots \cdot \exp\left[\mathcal{X}_{n}\right]\Ket{\delta^{n-1}\pi},\nonumber
\eqn
with 
\beqn
\mathcal{X}_i=\sum_{A,|A|=i}\tau_{A}\mathcal{R}_A,\nonumber
\eqn
and
\beqn
\mathcal{R}_{A}=\sum_{B\subseteq A}\fra{1}{2}\left(\mathfrak{L}_{\alpha}^{B}+\mathfrak{L}_{\beta}^{B}\right).\nonumber
\eqn
We will refer to this construction of a phylogenetic tensor as the ``$\mathfrak{L}$-representation''.
We will show that for a tree these two representations are exactly equal, but for arbitrary split systems this is not the case.

We will find it convenient to label the vector $\Ket{i_1i_2\ldots i_n}$ by the subset $A\subset \left[n\right]$ defined by setting $j\in A$ if and only if $i_j=1$. 
For example if $n=6$, we have
\beqn
\Ket{\emptyset}&=\Ket{000000},\\
\Ket{\left[6\right]}=\Ket{\{1,2,3,4,5,6\}}&=\Ket{111111},\\
\Ket{\{2,3,4\}}&=\Ket{011100},\\
\Ket{\{5\}}&=\Ket{000010}.\nonumber
\eqn

\begin{lem}\label{lem:equiv}
If 
\begin{itemize}
\item[(a.)] $A\subseteq B$, it follows that $\mathcal{R}_A\Ket{B}=\Ket{B-A}-\Ket{B},$
\item[(b.)] $A\cap B=\emptyset$, it follows that $\mathcal{R}_A\Ket{B}=\Ket{B\cup A}-\Ket{B}$.
\end{itemize}
If either case (a.) or (b.), we have
\beqn
\mathcal{R}_A\Ket{B}=\left(K^{(A)}-1\otimes 1\otimes \ldots \otimes 1\right)\Ket{B}.\nonumber
\eqn
\end{lem}
\begin{proof}
For all $A\subseteq B$ we have
\beqn
K^{(A)}\Ket{B}&=\Ket{B-A}.\nonumber
\eqn
On the other hand we know that
\beqn
\mathcal{R}_{\left[n\right]}\Ket{\emptyset}&=\Ket{\left[n\right]}-\Ket{\emptyset},\nonumber\\
\mathcal{R}_{\left[n\right]}\Ket{\left[n\right]}&=\Ket{\emptyset}-\Ket{\left[n\right]}.
\eqn
For $A\subseteq B\subseteq \left[n\right]$ it is always possible to permute tensor factors to write
\beqn
\Ket{B}\cong \Ket{\left[n_A\right]}\otimes \Ket{B'},\nonumber
\eqn
where $n_A=|A|$ and $B'\subset \left[n\right]-A$ with $B':=B-A$.
Thus
\beqn
\mathcal{R}_A\Ket{B}&\cong \left(\mathcal{R}_{\left[n_A\right]}\Ket{\left[n_A\right]}\right)\otimes \Ket{B'}\nonumber\\
&=\left(\Ket{\emptyset}-\Ket{\left[n_a\right]}\right)\otimes \Ket{B'}\\
&\cong \Ket{B-A}-\Ket{\left(B-A\right)\cup A}\\
&=\Ket{B-A}-\Ket{B},
\eqn
which proves the lemma for $A\subseteq B$.
The $A\cap B=\emptyset$ case follows from a similar argument.
\end{proof}

For two subsets $A,B$ taken from a compatible split system with $|A|=|B|$ it is the case that $A\cap B=\emptyset$, which in turn implies that $\left[\mathcal{R}_A,\mathcal{R}_B\right]=0$.
Thus we can make the replacement 
\beqn
\exp\left[\mathcal{X}_i\right]:=\exp\left[\sum_{A,|A|=i}\tau_{A}\mathcal{R}_{A}\right]=\prod_{A,|A|=i}\exp\left[\tau_A\mathcal{R}_A\right],\nonumber
\eqn
where by commutivity the product can be ordered in any way we please. 

Thus, in the $\mathfrak{L}$-representation, we see that for a tree we can write
\beqn
P=\ldots \exp\left[\tau_{A_3} \mathcal{R}_{A_3}\right]\exp\left[\tau_{A_2} \mathcal{R}_{A_2}\right]\exp\left[\tau_{A_1} \mathcal{R}_{A_1}\right]\Ket{\delta^{n-1}\pi},\nonumber
\eqn
with, due to the fact we are dealing with a tree, \emph{either} $A_i\cap A_{i+1}=\emptyset$ \emph{or} $A_{i+1}\subset A_{i}$.
Noting that $\Ket{\delta^{n-1}\pi}=\pi_0\Ket{\emptyset}+\pi_1\Ket{\left[n\right]}$ and repeated application of Lemma~\ref{lem:equiv} then gives:

\begin{thm}
For compatible split systems, the ``$\mathfrak{L}$-representation'' and the ``$K$-representation'' give rise to identical phylogenetic tensors.
\end{thm}

For \emph{arbitrary} split systems, however, this is not true, as the example in the next section shows.

\section{Phylogenetic networks and ``epochs''}\label{sec:epochs}

Consider the three taxa phylogenetic tree given in Figure~\ref{fig:3taxa}~(a).
As was proved above, if we take the binary symmetric model we get an identical probability distribution if we use either the $\mathfrak{L}$-representation $P_\ell=\exp\left[\tau_1\mathcal{R}_1+\tau_2\mathcal{R}_2+\tau_3\mathcal{R}_3\right]\cdot\exp\left[\tau_{12}\mathcal{R}_{12}\right]\Ket{\delta^2 \pi}$ or the $K$-representation $P_k=e^{-\lambda}\exp\left[\tau_1K^{(1)}+\tau_2K^{(2)}+\tau_3K^{(3)}+\tau_{12}K^{(12)}\right]\Ket{\delta^2 \pi}$, with $\lambda=\tau_1+\tau_2+\tau_3+\tau_{12}$.

We would like to introduce the additional parameter $\tau_{23}$ associated with the ``imaginary'' split $1|23$ to these probability distributions.  
We will do this in a way which is consistent with the design given in Figure~\ref{fig:3taxa}~(b), where the evolutionary history in broken up into three epochs: a. divergence of taxa 3 away from 1 and 2, b. convergent evolution of taxa 2 and 3, with independent divergence of taxa 3, and c. independent divergence of all taxa.

\begin{figure}[tbp]
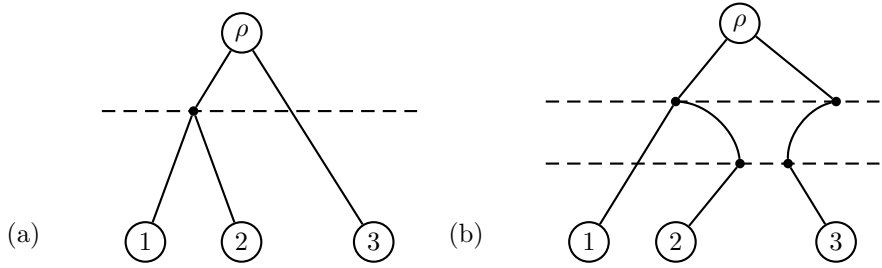

\centering
\label{fig:3taxa}

\begin{tabular}{cc}
(a)\hspace{2em}
$
\psmatrix[colsep=.3cm,rowsep=.7cm,mnode=circle,arrowscale=2]
&&&\rho\\
[mnode=dot,dotscale=.00001]&&[mnode=dot,dotscale=1]&&&&&&[mnode=dot,dotscale=.00001]\\
[mnode=dot,dotscale=.00001]\\
&1&&2&&&&3
\ncline{1,4}{2,3}
\ncline{1,4}{4,8}
\ncline{2,3}{4,2}
\ncline{2,3}{4,4}
\psset{linestyle=dashed}\ncline{2,1}{2,9}
\endpsmatrix
$ 

&
(b)\hspace{2em} 
$
\psmatrix[colsep=.3cm,rowsep=.7cm,mnode=circle,arrowscale=2]
&&&&\rho\\
[mnode=dot,dotscale=.00001]&&&[mnode=dot,dotscale=1]&&&[mnode=dot,dotscale=1]&[mnode=dot,dotscale=.00001]\\
[mnode=dot,dotscale=.00001]&&&&[mnode=dot,dotscale=1]&[mnode=dot,dotscale=1]&&[mnode=dot,dotscale=.00001]\\
&1&&2&&&3
\ncline{1,5}{2,4}
\ncline{1,5}{2,7}
\ncline{2,4}{4,2}
\ncarc[arcangle=40]{2,4}{3,5}
\ncarc[arcangle=-40]{2,7}{3,6}
\ncline{3,5}{4,4}
\ncline{3,6}{4,7}
\psset{linestyle=dashed}\ncline{2,1}{2,8}
\psset{linestyle=dashed}\ncline{3,1}{3,8}
\endpsmatrix
$ 
\\
\end{tabular}

\caption{A three taxa tree (a.) is modified by the introduction of the addtional split $\{23\}$ in (b).}
\end{figure}

To model this situation we must introduce the additional edge to each representation. 
To make the presentation as simple as possible, we set a molecular clock on the model such that $\tau_2=\tau_1$ and $\tau_3=\tau_{12}+\tau_{1}$, and we introduce a scaling parameter $\theta\in\left[0,1\right]$ to control the length of the second epoch as a proportion of the third epoch by setting $\tau_{23}=\tau_1\theta$ .
As all operators in the $K$-representation commute, the only choice available in this case is to take
\beqn
P'_K=e^{-\lambda}\exp\left[\tau_1K^{(1)}+\tau_1(1\!-\!\theta)K^{(2)}+\left(\tau_{12}+\tau_1(1\!-\!\theta)\right)K^{(3)}+\tau_{12}K^{(12)}+\tau_{1}\theta K^{(23)}\right]\Ket{\delta^2 \pi}.\nonumber
\eqn
This is exactly consistent with the generalizations given in \citet{bryant2005,bryant2009}.

For the $\mathfrak{L}$-representation we do \emph{not} have commutivity of the operators $\mathcal{R}_2,\mathcal{R}_3,\mathcal{R}_{12}$ with the new operator $\mathcal{R}_{23}$, thus we need to proceed more carefully as there is some choice in how the extra edge is introduced.
Using the diagram and its three epochs as a guide, we take
\beqn
P'_{\mathfrak{L}}=\exp\left[\tau_1(1\!-\!\theta)\left(\mathcal{R}_1+\mathcal{R}_2+\mathcal{R}_3\right)\right]\cdot\exp\left[\tau_{1}\theta\left(\mathcal{R}_1+\mathcal{R}_{23}\right)\right]\cdot\exp\left[\tau_{12}\left(\mathcal{R}_3+\mathcal{R}_{12}\right)\right]\Ket{\delta^2 \pi}.\nonumber
\eqn
For clarity of comparison, we write the $K$-representation in epoch form:
\beqn
P'_K=e^{-\lambda}\exp\left[\tau_1(1\!-\!\theta)\left(K^{(1)}+K^{(2)}+K^{(3)}\right)\right]\cdot&\exp\left[\tau_{1}\theta\left(K^{(1)}+K^{(23)}\right)\right]\\
&\cdot\exp\left[\tau_{12}\left(K^{(3)}+K^{(12)}\right)\right]\Ket{\delta^2 \pi}.\nonumber
\eqn

Now consider the state of the probability distribution at the beginning of epoch 2. 
As we are dealing with the binary symmetric model, it is clear that the probability of the any state $\Ket{ijk}$ is invariant to permutation of the states $0\rightleftharpoons1$.
Also, the structure of the tree up to the start of epoch 2 implies that the probability of any state of the form $\Ket{ijk}$ where $i\neq j$ is of probability zero.
Thus we can assume at the start of epoch 2 that the distribution is of the form
\beqn
P=(1-q)\fra{1}{2}\left(\Ket{000}+\Ket{111}\right)+q\fra{1}{2}\left(\Ket{001}+\Ket{110}\right),\nonumber
\eqn
where, because we are considering a continuous-time model, $q\in\left[0,\fra{1}{2}\right)$.

Considering the definitions
\beqn
\mathcal{R}_{23}=1\otimes\left(L_\alpha\otimes L_\alpha+L_\alpha\otimes 1+1\otimes L_\alpha+L_\beta\otimes L_\beta+L_\beta\otimes 1+1\otimes L_\beta\right),\nonumber
\eqn
and
\beqn
K^{(23)}=1\otimes K\otimes K=\mathcal{R}_{23}+1\otimes\left(L_\alpha\otimes L_\beta+L_\beta\otimes L_\alpha\right),\nonumber
\eqn
it follows that transition rates between the four existing states in the two cases are given by the two graphs in Figure~\ref{fig:ratediag}, where all transition rates are equal.
The crucial thing to note is that $\mathcal{R}_{23}$ ``corrects'' patterns that are inconsistent with the split $1|23$, whereas $K^{(23)}$ simply \emph{premutes} these two states.

\begin{figure}
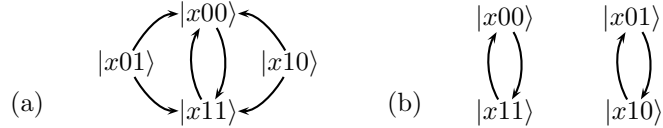

\centering
\label{fig:ratediag}
\begin{tabular}{cccc}
(a)\hspace{2em}$
\psmatrix[colsep=.3cm,rowsep=.3cm]
&\Ket{x00}&\\
\Ket{x01}&&\Ket{x10}\\
&\Ket{x11}&
\ncarc[arcangle=30]{<-}{3,2}{2,1}
\ncarc[arcangle=30]{->}{2,1}{1,2}
\ncarc[arcangle=-30]{<-}{3,2}{1,2}
\ncarc[arcangle=-30]{<-}{1,2}{3,2}
\ncarc[arcangle=-30]{<-}{3,2}{2,3}
\ncarc[arcangle=30]{<-}{1,2}{2,3}
\endpsmatrix
$
&\hspace{4em}&
(b)\hspace{2em}$
\psmatrix[colsep=.9cm,rowsep=.9cm]
\Ket{x00}&\Ket{x01}\\
\Ket{x11}&\Ket{x10}
\ncarc[arcangle=30]{->}{1,1}{2,1}
\ncarc[arcangle=30]{->}{2,1}{1,1}
\ncarc[arcangle=30]{->}{1,2}{2,2}
\ncarc[arcangle=30]{->}{2,2}{1,2}
\endpsmatrix
$
\end{tabular}
\caption{Transition rates for states $\Ket{xyz}$ under (a.) $\mathcal{R}_{23}$, and (b.) $K^{({23})}-1\otimes 1\otimes 1$.}
\end{figure}

It is now clear the there is quite a marked difference between the $\mathfrak{L}$- and $K$-representations. 
The $\mathfrak{L}$ representations introduces a natural notion of the ``coming together'' of taxa. 
In fact it is easy to see directly from the diagram that in the limit of extension of the edge $\tau_{23}$ to infinity that the probability distribution will converge to 
\beqn
P&=\left(\left(1-\fra{1}{2}q\right)\fra{1}{2}+\fra{1}{4}q\right)\left(\Ket{000}+\Ket{111}\right)+\fra{1}{2}q\fra{1}{2}\left(\Ket{011}+\Ket{100}\right)\\
&=\left(1-\fra{1}{2}q\right)\fra{1}{2}\left(\Ket{000}+\Ket{111}\right)+\fra{1}{2}q\fra{1}{2}\left(\Ket{011}+\Ket{100}\right),\nonumber
\eqn
which is consistent with the probability distribution that would arise under a tree where taxa 1 has diverged from 2 and 3, but there has been zero divergence of taxa 2 and 3 themselves.
This behaviour is of the course the reasoning behind the way we have chosen to draw our diagram Figure~\ref{fig:3taxa}~(b).

The $K$-representation cannot achieve this type of convergence, with its limiting state being
\beqn
(1-q)\fra{1}{4}\left(\Ket{000}+\Ket{011}+\Ket{111}+\Ket{100}\right)+q\fra{1}{4}\left(\Ket{001}+\Ket{010}+\Ket{110}+\Ket{101}\right).\nonumber
\eqn

%The $\mathcal{L}$-representation seems to be superior in the sense that it has a natural notion of bringing edges back together.
%However, the $\mathcal{L}$-rep has serious marginilization issues and who knows about identifiability.
%\input{sec3}
\section{Conclusion}
In this paper we have shown how to express the two state continuous-time general Markov model on trees in such a way that allows extension to arbitrary split systems and even more general network models.
By reviewing the lemmas in Section~\ref{sec:prelim} it is clear that the results extend easily to the general Markov model with more character states.
However, we defer confirmation of this observation to future work.

In Section~\ref{sec:incon} we showed that our discussion gives rise to phylogenetic models for the general Markov model that are identical to previous approaches only on \emph{trees}, and in Section~\ref{sec:epochs} we give a simple example that shows how our approach allows for convergent evolution of previously divergent lineages (a structural property that was previously unobtainable).

Besides its theoretical interest, we expect that the ability to model convergent evolution in this way will have a significant application where it is known that particular datasets exhibit non-treelike behaviour due to population genetic properties such as incomplete lineage sorting and other effects that confound strictly treelike models (as discussed in Section~\ref{sec:intro}). 
We suspect that exploration of the relation between the network models that arise in our discussion and simple models of population genetics is likely to yield significant additional insight.   

Finally, comparison of our network models to the distribution space generated by the so-called ``mixture models'' \citep{pagel2004} is also in need of investigation.
For example, comparison of a mixture of the trees $12|3$ and $1|23$ to our network example given in Section~\ref{sec:epochs} should yield significant theoretical insight into the biological meaning and plausibility of these differing model classes.  
Careful scientific thought is required to tease out what biological processes are explicitly (or implicitly) being modelled by either of these approaches.

Of course these exciting possibilities must be tempered by analysis of the \emph{identifiability} (or otherwise) of the models that arise by taking more general networks. 
Establishing identifiability is not only essential from a statistical inference point of view, but, in this case, may lead to natural restrictions of the types of network that can be realistically used for phylogenetic inference.

\bibliographystyle{jtbnew}
\bibliography{masterAB}

\end{document}